\pdfoutput=1
\documentclass[%
 reprint,
 superscriptaddress,
 showpacs,
 amsmath,amssymb,
 aps,
 pra,
]{revtex4-2}

\usepackage{graphicx}
\usepackage{dcolumn}
\usepackage{bm}
\usepackage{comment}
\usepackage[colorlinks,
linkcolor = blue,
anchorcolor = blue,
urlcolor = blue,
citecolor = blue
]{hyperref}
\usepackage[mathlines]{lineno}

\usepackage{amsmath,amsthm}
\usepackage{amssymb}

\newtheorem{thm}{Theorem}
\newtheorem{pro}{Proposition}
\newtheorem{lem}{Lemma}

\newtheorem{conj}{Conjecture}

\newcommand{\tr}{\mathrm{tr}}

\newcommand{\E}{\mathcal{E}}
\newcommand{\I}{\mathcal{I}}

\newcommand{\y}{\rangle}
\newcommand{\yy}{\rangle \! \rangle}
\newcommand{\zz}{\langle \! \langle}
\newcommand{\1}{\boldsymbol{1}}
\makeatletter
\newcommand{\TODO}[1]{\textcolor{red}{\@ifnotempty{#1}{ #1}}}
\makeatother
\begin{document}


\title{Relating measurement disturbance, information and orthogonality}

\author{Yizhou Liu}
 \email{liuyz18@mails.tsinghua.edu.cn}
\affiliation{
 Department of Engineering Mechanics, 
 Tsinghua University,
 Beijing 100084, China
}

\author{John B. DeBrota}
\affiliation{
Department of Physics and Astronomy, Tufts University, 574 Boston Avenue, Medford MA 02155, USA
}
\date{\today}

\begin{abstract}
In the general theory of quantum measurement, one associates a positive semidefinite operator on a $d$-dimensional Hilbert space to each of the $n$ possible outcomes of an arbitrary measurement. In the special case of a projective measurement, these operators are pairwise Hilbert--Schmidt orthogonal, but when $n>d$, orthogonality is restricted by positivity. This restriction allows us to more precisely state the quantum adage: information gain of a system is always accompanied by unavoidable disturbance. Specifically, we investigate three properties of a measurement with L\"uders rule updating: its disturbance, a measure of how the expected post-measurement state deviates from the input; its measurement strength, a measure of the intrinsic information producing capacity of the measurement; and its orthogonality, a measure of the degree to which the measurement operators differ from an orthonormal set. These quantities satisfy an information-disturbance trade-off relation that highlights the additional role played by orthogonality. Finally, we assess several classes of measurements on these grounds and identify symmetric informationally complete quantum measurements as the unique quantum analogs of a perfectly informative and nondisturbing classical ideal measurement.
\end{abstract}

\maketitle

\section{\label{sec:intro}introduction}

Well-established quantum lore states that information gain about a system from measurement comes at the cost of an inescapable finite disturbance~\cite{FP96,F98,Ban01,FJ01,Bar02,Mac06,BS06,Luo10,SKU16,ZZY16,BGD+20}. Generally, the more informative the measurement, the larger the disturbance. This is contrasted with an imagined classical picture where the complete physical condition of a system is captured by a set of coordinates which may be passively observed; classically, one may obtain complete information with zero disturbance. What can this difference tell us about the nature of reality? What can it tell us about the informational power of quantum theory?

In this paper, we present a quantification of the information-disturbance trade-off which highlights the role played by the orthogonality, or lack thereof, of the operators associated with a quantum measurement.

First, we introduce some notations and definitions.
Let $\mathcal{H}_d$ denote a $d$-dimensional Hilbert space, and let $L(\mathcal{H}_d)$ denote the $d^2$-dimensional real vectorspace of Hermitian operators
on $\mathcal{H}_d$, equipped with the Hilbert--Schmidt inner product $\tr (A^{\dagger}B)$. To a system, one assigns a quantum state, also called a density operator, which is a positive semidefinite unit trace operator in $L(\mathcal{H}_d)$. To a measurement $\E$, one assigns a positive operator valued measure (POVM) $\{E_i\}$, which is a set of $n$ positive semidefinite operators in $L(\mathcal{H}_d)$ which sum to the identity. Each operator $E_i$, called an effect, is associated with a possible outcome of the measurement. 

Effects may be written as rescalings of unit trace operators,
\begin{equation}
E_i = e_i \Pi_i,~~e_i = \tr E_i\;.
\end{equation}
If all the weights $e_i$ are equal, the POVM is called unbiased, and if all $\Pi_i$s are rank-1, we say the POVM is rank-1. 

Given a quantum state $\rho$ and a POVM $\{E_i\}$, the probability for the $i$th outcome is given by the Born rule:
\begin{equation}\label{born}
p_i = \tr E_i\rho\;.
\end{equation}

While a POVM allows one to assign probabilities to possible outcomes, the post-measurement state depends on the specific interaction 
mode~\cite{FJ01}. The L\"uders 
rule \cite{Bar02,BL09} provides a distinguished choice.
For an input state $\rho$ and measurement outcome $i$, the L\"uders rule post-measurement state is
\begin{equation}\label{luders}
\rho_i=\frac{1}{p_i}\sqrt{E_i}\rho\sqrt{E_i}\;,
\end{equation}
where $\sqrt{E_i}$ denotes the unique positive semidefinite square root of $E_i$. 

Before performing a measurement, one may consider the \emph{expected} post-measurement state of a measurement $\E$, that is, the convex mixture of the possible post-measurement states weighted by the Born rule probabilities of each outcome. For L\"uders rule updating, the expected post-measurement state is given by the channel
\begin{equation}
\E(\rho) = \sum_i p_i \rho_i=\sum_i \sqrt{E_i}\rho\sqrt{E_i}\;,
\label{post}
\end{equation}
where $\E$ alternately refers to the measurement or the channel when it is clear from context. 

The L\"uders channel is unique among possible interaction modes which do not break quantum coherence more than is necessary~\cite{Bar98,Ban01,FJ01}. As we wish to compare the disturbing tendencies of different measurements, in the following analysis we only consider the
L\"uders channel interaction mode, and we use it to derive a characteristic
measure for a POVM's intrinsic
disturbance. We refer to the set $\{p_i,\rho_i\}$ given by Eq.~\eqref{born} and Eq.~\eqref{luders} as the characteristic post-measurement ensemble for a measurement $\E$.

The most familiar quantum measurements are projective measurements, also called von Neumann measurements. The effects of a projective measurement are the pairwise orthogonal projectors onto the eigenspaces of a Hermitian ``observable''. A projective measurement may have as many as $n=d$ outcomes if the eigenvalues of the observable are distinct. In this case, the POVM will be rank-1 and unbiased. L\"uders rule is implicit in a nondegenerate projective measurement; the post-measurement state is the eigenprojector associated with the outcome obtained.

Projective measurements are severely limited by $n\leq d$. A general POVM may have any number of outcomes, unrestricted by $d$. A consequence of this is that the POVM associated to a general measurement may span more of $L(\mathcal{H}_d)$. In fact, a POVM may contain as many as $d^2$ linearly independent effects, spanning the entirety of $L(\mathcal{H}_d)$. Such POVMs are called informationally complete (IC) because the outcome probabilities $p_i$ are in one-to-one correspondence with the quantum state $\rho$~\cite{Busch91,Pru77}. IC-POVMs with the minimal number of effects, $n=d^2$, are called minimal IC-POVMs (MICs). MICs and IC-POVMs exist and may be constructed in every finite dimension~\cite{DFS20}.

Another important class of measurements are the equiangular POVMs (EA-POVMs). For us, an EA-POVM is the rank-1 and unbiased POVM constructed from an equiangular tight frame~\cite{Wal18}. An EA-POVM is equiangular meaning that the Hilbert--Schmidt inner product between all pairs of distinct effects is equal. The effects of an EA-POVM $\{E_i \}_{i=1}^{n}$ satisfy
\begin{eqnarray}
\tr \Pi_i \Pi_j = \frac{1}{d} \frac{n-d}{n-1}~(i \neq j)\;
,~~\text{and}~~
e_i = \frac{d}{n}~(\forall i)\;,
\end{eqnarray}
where $d \leq n \leq d^2$. Note that EA-POVMs are not guaranteed to exist for a given $n \in [d,d^2]$~\cite{Fer14}.

EA-POVMs and IC-POVMs can only coincide when $n=d^2$, as this is the minimal number of outcomes to be IC and the maximal number of outcomes to be equiangular. An EA-POVM with $n=d^2$ is a special kind of MIC called a  
symmetric informationally complete POVM 
(SIC)~\cite{RBSC04,Scott06,AFZ15,SIC17}.
SICs are also distinguished by having the minimal number of effects among POVMs that can form a 2-design~\cite{RBSC04,Scott06}.

The remainder of the paper is organized as follows.
Sec.~\ref{sec:dist} exhibits the intuition and basic properties
of our disturbance measure $D$. We show that disturbance may be written as a sum of terms involving the POVM's effects, which we interpret in the following two sections. In Sec.~\ref{sec:reso}, we motivate and define a measurement's measurement strength $R$, a measure of the expected information gain from measurement based on a particular generalized entropy. In Sec.~\ref{sec:orth} we define a measurement's orthogonality $O$, a measure of the degree to which its effects differ from an orthonormal set. 
We show that for a POVM $\E$ with $n\;(d\leq n \leq d^2)$ effects,
its orthogonality is restricted, namely
\begin{eqnarray}
O(\E) \geq \frac{(n-d)^2}{n-1}\;,
\end{eqnarray}
and the lower bound can be achieved by EA-POVMs. Finally,
disturbance $D$ can be decomposed as follows,
\begin{eqnarray}
D = 2d(R+1) -d^2 - n + O\;.
\label{eq}
\end{eqnarray}
This relation is examined in Sec.~\ref{sec:dis}, where we also discuss the conceptual content of our work, especially as it relates to a hypothetical perfectly informative and nondisturbing classical measurement.
Finally, we conclude with some brief outlooks in Sec.~\ref{app:concl}.

\section{\label{sec:dist}Disturbance}

We quantify a measurement's disturbance by the difference between an input state and the expected post-measurement state, assuming L\"uders rule updating. 

The expected post-measurement state for the input $\rho$ and POVM $\{E_i\}$ is given by the L\"uders channel Eq.~\eqref{post}. 
The Hermitian operators can be regarded as vectors as $L(\mathcal{H}_d)$
is also a linear space. Let $|\rho\yy$ denote the vector form of the operator $\rho$.
The L\"uders channel $\E$ is equivalent to a linear operator on $|\rho\yy$, namely, $|\E(\rho)\yy = [\E]|\rho\yy$. Here, 
$[\E]$ is the linear operator corresponding to $\E$.
The 2-norm of vectors $|\rho\yy$ is equivalent to the Hilbert-Schmidt 
norm of operators $\rho$: $\||\rho\yy \|_2^2 = \zz\rho | \rho\yy = \tr\rho^2$.
We now use the square of the Hilbert-Schmidt distance to quantify the difference between $\rho$ and $\E(\rho)$, which can be expressed by
\begin{equation}
    D(\rho,\E(\rho)) :=\||[\E]|\rho\yy -|\rho\yy\|^2_2 = \tr(\E(\rho)-\rho)^2\;.
\end{equation}
If $\I$ is the identity channel, $\I(\rho) = \rho$ and $[\I]|\rho\yy = |\rho\yy$. Based on the theory of operator norms, 
the Frobenius norm $\|[\E]-[\I]\|^2_F$ is an upper bound of $D(\rho,\E(\rho))$.
The difference between $[\E]$ and the identity $[\I]$ thus provides an operator characterization of the property we are after without initial state dependence.
We define $\E$'s intrinsic disturbance as
\begin{equation}
D(\E) := \|[\E] - [\I] \|^2_F \equiv \tr([\E]-[\I])^2\;.
\end{equation}

To examine $D(\E)$, we first find 
the explicit matrix form of the operator $[\E]$.
Let $\{X_i \}$, $i=1,\ldots,d^2$, be an orthonormal basis for $L(\mathcal{H}_d)$. Then $[\E]$ can be expressed as
\begin{equation}
([\E])_{ij} = \sum_{l=1}^n \tr (X_i \sqrt{E_l} X_j \sqrt{E_l})\;.
\end{equation}
In the basis $\{X_i \}$, we have
\begin{eqnarray}\label{disturbancecommutator}
([\E-\I])_{ij} &&= \sum_l \tr (X_i \sqrt{E_l} X_j \sqrt{E_l})-\delta_{ij}
\notag
\\
&&=\frac{1}{2} \sum_l \tr ([\sqrt{E_l},X_i][\sqrt{E_l},X_j])\;,
\end{eqnarray}
where $[A,B]=(AB-BA)$ denotes the commutator of $A$ and $B$.
The latter form of Eq.~\eqref{disturbancecommutator} shows that disturbance is strongly
related to non-commutativity.
If $D(\E) = 0$, we have that $([\E - \I])_{ii} = 0$ for all $i$.
From this, it follows that $\sqrt{E_l}$ commutes with all
Hermitian operators, implying that $\sqrt{E_l}$ is proportional to the identity.
The following result is readily concluded

\begin{pro}\label{thm1}
$D(\E) \geq 0$,
with equality iff all effects 
of $\E$ are proportional to 
the identity (i.e., $\E = \I$).
\end{pro}

Since $D(\cdot)$ is free of the choice of 
orthonormal basis of $L(\mathcal{H}_d)$,
it is convenient for calculations to use the basis $\{A_{\mu\nu}\}$, where
\begin{eqnarray}
A_{\mu \nu}  \equiv 
\left\{
\begin{array}{lll}
| \mu \rangle\langle \mu |,~\mu = \nu
\\
\frac{1}{\sqrt{2}} (| \mu \rangle
\langle \nu | + | \nu \rangle \langle \mu |),
~\mu < \nu
\\
\frac{i}{\sqrt{2}} (| \mu \rangle
\langle \nu | - | \nu \rangle \langle \mu |),
~\mu > \nu\;,
\end{array}
\right.
\end{eqnarray}
Explicit computation reveals the following expression for disturbance in terms of the effects of the POVM:
\begin{widetext}
\begin{eqnarray}
D(\E) &&= 
\sum_{l,m,a,b,c,d}\tr(\sqrt{E_l}A_{ab}\sqrt{E_l}A_{cd})
\tr(\sqrt{E_m}A_{cd}\sqrt{E_m}A_{ab})
-2\sum_{l,a,b} \tr(\sqrt{E_l}A_{ab}\sqrt{E_l}A_{ab}) + 
d^2
\notag
\\
&&=\sum_{l,m,a,b,c,d}(\sqrt{E_l})_{ab}(\sqrt{E_m})_{ba}
(\sqrt{E_l})_{cd}(\sqrt{E_m})_{dc}
-2\sum_{l,a,b}(\sqrt{E_l})_{aa}(\sqrt{E_l})_{bb}
+d^2
\notag
\\
&&=\sum_{l,m}\tr^2(\sqrt{E_l}\sqrt{E_m})
-2\sum_l \tr^2(\sqrt{E_l})
+d^2\;.
\label{D}
\end{eqnarray}
\end{widetext}
The terms $\sum_l \tr^2(\sqrt{E_l})$ and $\sum_{l,m}\tr^2(\sqrt{E_l}\sqrt{E_m})$ 
in Eq.~\eqref{D} separately appear in two other characteristic properties of measurements which we introduce and investigate in the next two sections. 

\section{\label{sec:reso}Information}
The term $\sum_l \tr^2(\sqrt{E_l})$ appearing in the 
disturbance expression Eq.~\eqref{D} is related to the capacity of a measurement with L\"uders updating rule to produce a particular kind of invariant information.

Information gain is often conceived in terms of a reduction of entropy. The most common quantum entropy is the von Neumann entropy $S(\rho):=-\tr \rho \ln \rho$. Although a distinguished choice, von Neumann entropy is not the only measure available. We make use of the generalized 
entropies of a state $\rho$, given
by~\cite{PI91,HY06} 
\begin{eqnarray}
S_{s}^{(\alpha)}(\rho) :=  
\frac{\tr^s(\rho^{\alpha})-1}{s(1-\alpha)}\;,
\end{eqnarray}
where $s$ is a parameter identifying
families of entropies, each member of which is labeled by $\alpha$.
The R\'enyi entropies are the limiting $s\to 0$ family and, within these, the von Neumann entropy is obtained by taking $\alpha \to 1$.

Depending on one's purpose, any of the generalized entropies may be suitable for quantifying uncertainty in a quantum system. Two notable examples are $S^{(2)}_1$ and $S^{(1/2)}_2$. The linear entropy $S^{(2)}_1$ is complementary to the Brukner--Zeilinger information \cite{BZ01,BZ99,BZ09} and was shown by Luo to be connected to the variance of an observable~\cite{Luo07}. In \cite{Luo06}, Luo derived another invariant information by replacing the variance with the Wigner--Yanase skew information \cite{WY63}. This information measure is complementary to $S^{(1/2)}_2$ and is given by
\begin{equation}
    I_{\rm WY}(\rho):=d-\tr^2\sqrt{\rho}\;.
\end{equation}

The expected gain in this quantity resulting from a measurement $\E$ on a state $\rho$ is given by 
\begin{equation}
    R(\rho,\E):=\sum_i p_iI_{\rm WY}(\rho_i)-I_{\rm WY}(\rho)\;.
\end{equation}
For more context on this information measure, see Appendix \ref{sec:appa}.

To quantify the intrinsic information producing capacity of the measurement alone, we consider the quantity that arises from inputting the state of zero information, the maximally mixed state $\rho=I/d$. Then we have
\begin{eqnarray}
R(\E) :=R(I/d,\E)
=d -\frac{1}{d} \sum_i\tr^2
\sqrt{E_i}\;,
\label{R2}
\end{eqnarray}
which features the term of interest, $\sum_l \tr^2(\sqrt{E_l})$.

The following is easily verified:

\begin{pro}\label{thm2}
For any POVM, $R(\E)$ has lower and upper bounds
\begin{eqnarray}
0 \leq R(\E) \leq d-1\;.
\end{eqnarray}
The lower bound is saturated iff
all effects are proportional to the identity,
and the upper bound is saturated iff all effects are rank-1.
\end{pro}

The closer effects are to being proportional to the identity, the less information of this kind the measurement produces. On the other hand, if the effects are close to being rank-1, the expected information gain is high. Following a similar line of reasoning as in Ref.~\cite{DJJ01}, we call $R(\E)$ the measurement strength of $\E$.

\section{\label{sec:orth}Orthogonality}
Now we turn to the other term in Eq.~\eqref{D}, $\sum_{l,m}\tr^2(\sqrt{E_l}\sqrt{E_m})$. When $l\neq m$, nonzero contributions occur when $\sqrt{E_l}$ and $\sqrt{E_m}$ are not Hilbert--Schmidt orthogonal. The effects of a projective measurement are pairwise orthogonal projectors, but for general measurements, some degree of overlap may occur. Thus, we find that this term of the disturbance is related to the degree to which the effect operators deviate from orthogonality. 

Consider the Gram matrix of $\{\sqrt{E_i}\}$,
\begin{eqnarray}
[S]_{ij} := \tr(\sqrt{E_i}\sqrt{E_j})\;.
\end{eqnarray}

The effects of a projective measurement with $n=d$ distinct outcomes are orthonormal projectors, so $S$ equals the $n\times n$ identity matrix $\1$. When $n>d$, the requirement that effects be positive semidefinite prevents $S$ from equaling $\1$ within quantum theory. 

To quantify this distinction, we introduce the orthogonality $O(\E)$ for a POVM $\E=\{E_i \}_{i=1}^{n}$
\begin{eqnarray}
O(\E) := \|S -\1 \|^2_F = \tr(S -\1)^2\;,
\end{eqnarray}
which \emph{decreases} when effects are closer to orthonormal. When $O(\E)=0$, the effects are perfectly orthonormal; the farther away they are from this limit, the larger $O$ becomes.
Elementary algebra gives 
\begin{eqnarray}
O(\E) = \sum_{l,m}\tr^2(\sqrt{E_l}\sqrt{E_m}) - 2d + n\;,
\label{O}
\end{eqnarray}
the first term of which coincides with Eq.~\eqref{D}.

If there are too many 
vectors to form an orthogonal basis,
the circumstance most similar to an orthonormal basis
should be an equiangular set, if $n$ is such that this is possible.
Indeed, the following theorem states that for a fixed $d\leq n \leq d^2$, EA-POVMs minimize $O$.

\begin{thm}\label{thm3}
For any unitarily invariant norm $\|\cdot \|$ and real parameter $\gamma$,
\begin{eqnarray}
\|S_{n \rm EA} - \gamma\1 \| \leq 
\|S_{n} - \gamma\1 \|\;,
\label{T3}
\end{eqnarray}
where $S_n$ is the $S$ matrix of a 
POVM with $n$ nonzero effects
and $S_{n \rm EA}$ the $S$ matrix of an EA-POVM with 
$n$ effects.
\end{thm}

The proof of Theorem \ref{thm3} is given in Appendix \ref{app:proof}. Letting $\gamma=1$ in Theorem \ref{thm3}, it is an immediate consequence
that $O(\E_{n \rm EA}) \leq O(\E_n)$\;,
and that
\begin{eqnarray}
O(\E_{n\rm EA}) = \frac{(n-d)^2}{n-1}\;.
\label{oea}
\end{eqnarray}
This lower bound is general, even if there is no corresponding EA-POVM. 

We suspect the lower bound $\|S_{n \rm EA} - \gamma\1 \|$ can only be saturated by EA-POVMs, which can be rigorously expressed as follows: 

\begin{conj}\label{con1}
Fixing $n$, 
\begin{eqnarray}
\|S_{n\rm EA} - \gamma\1 \| = 
\|S_{n} - \gamma\1 \|\;,
\end{eqnarray}
for any real $\gamma$ only if $S_n$ is the $S$ matrix of an 
EA-POVM.
\end{conj}

When $n \notin [d,d^2]$, EA-POVMs certainty do not exist, and the exact lower bound of $O$ may not be so elegant and direct. We discuss this in the Appendix \ref{sec:appb}. Considering values $\gamma<1$ would enable comparisons not just to orthonormality, but to any degree of sub-normalized orthogonality, a case which may also be interesting. 

\section{\label{sec:dis}Discussion}

Combining the results we have, it is easy to show that 
disturbance $D$ satisfies the decomposition Eq.~\eqref{eq}, which we reproduce here:
\begin{eqnarray}\label{decomposition}
D = 2d(R+1) -d^2 - n + O\;.
\end{eqnarray}
This expression reveals that the disturbance of a measurement with L\"uders rule updating on a $d$ dimensional quantum system may be directly thought of 
in terms of its measurement strength, orthogonality and 
the number of measurement outcomes.

From Proposition \ref{thm1}, we know that $D=0$ if and only if all effects of the POVM are proportional to the identity. The probability for outcome $i$ for such a measurement simply equals $e_i$, the weight of the $i$th effect, irrespective of the quantum state of the system being measured. In other words, the nondisturbing limit of measurement in quantum theory is entirely informationless. 

When a measurement is not informationless, $D>0$. To gain intuition for the decomposition Eq.~\eqref{decomposition} in this case, let's see some explicit examples.

\emph{Example 1: von Neumann measurements.} For any nondegenerate von Neumann measurement on a $d$ dimensional system, the measurement strength is $R = d-1$ because each effect is rank-1,
and the orthogonality is $O=0$ because the effects are orthonormal. As $n=d$, the disturbance reduces to 
\begin{equation}
    D_{\rm vN}=d^2-d\;.
\end{equation}

\emph{Example 2: EA-POVMs.} As an EA-POVM is rank-1, the measurement strength is $R=d-1$ and we know the orthogonality from Eq.~(\ref{oea}). After algebraic simplification, the disturbance reduces to
\begin{eqnarray}\label{DnEA}
D_{n \rm EA} = \frac{n(d-1)^2}{n-1}\;.
\end{eqnarray}
When $n=d$, an EA-POVM is a von Neumann measurement. As $n$ grows, the disturbance of an EA-POVM, if one exists for a given $n$, decreases until $n=d^2$, corresponding to a SIC, where
\begin{eqnarray}
D_{\rm SIC} = \frac{d^2(d-1)}{d+1}\;.
\label{SICD}
\end{eqnarray}

\emph{Example 3: rank-1 MICs.} For a rank-1 MIC, as $R = d-1$ and $n=d^2$, disturbance is completely 
controlled by orthogonality: 
\begin{equation}
    D_{\rm r1MIC}=O_{\rm r1MIC}\;.
\end{equation}
As SICs have the minimal orthogonality in this class, they also have the minimal disturbance: $D_{\rm r1MIC}\geq D_{\rm SIC}$.

\emph{Example 4: 2-designs.} A POVM which forms a 2-design is rank-1 and has at least $d^2$ effects. When $n=d^2$, it forms a SIC. 
A corollary of Theorem 2 in Ref.~\cite{RBSC04} gives us the orthogonality
\begin{eqnarray}
O_{\rm 2d} = \frac{2d}{d+1} -2d + n\;,
\end{eqnarray}
from which it follows that $D_{\rm 2d}=D_{\rm SIC}$, that is, all 2-designs achieve the same disturbance. In the Appendix, we see this is the minimal possible disturbance among measurements with maximal measurement strength, a result similar to one in Ref.~\cite{Bar02}.

\emph{Example 5: Reflected SICs.} When a SIC exists in dimension $d$ with projectors $\Pi_i$, one may form an unbiased rank-$(d-1)$ MIC,
\begin{equation}
    E_i=\frac{1}{d(d-1)}(-\Pi_i+I)\;.
\end{equation}
This MIC might be called a reflected SIC because the effects are proportional to the points on the boundary of quantum state space one obtains by reflecting the projectors of a SIC through the maximally mixed state. When $d=2$, the reflected SIC is another SIC because it is rank-1, but this does not happen in higher dimensions. One may work out that the measurement strength is $R=1$, the orthogonality is
\begin{equation}
    O_{\rm rSIC}=\frac{2d^4-4d^3-d^2+4d}{d^2-1}\;,
\end{equation}
and the disturbance is
\begin{equation}
    D_{\rm rSIC}=\frac{d^2}{d^2-1}\;.
\end{equation}
The disturbance is quite low, achieved at the cost of a measurement strength of $1$, which falls shorter and shorter of the maximal $d-1$ as one considers larger Hilbert space dimensions.

What do our results teach us about the relation between quantum and classical mechanics? To address this, we can think about classical versions for each of the quantities we considered quantum mechanically: disturbance, measurement strength, and orthogonality. How and why quantum mechanics satisfies different constraints may help us guess what kind of natural pressures make it the right theory to use in our universe.

A common way to conceive of classicality is as a commuting quantum subtheory where classical properties are revealed by the appropriate von Neumann measurement. However, our considerations of disturbance and Refs.~\cite{DFS20a} and \cite{DFS20} inspire us to follow another line of reasoning motivated by the way one may interpret informational completeness in the two theories. 

While we can learn nothing from a zero disturbance measurement in quantum mechanics, this is not true classically. From a classical perspective, measurement is understood in terms of revealing preexisting properties of a system, so when we speak of disturbance from a classical measurement, we are supposing the action of measurement will cause a change to these properties. One can thus imagine an informationally complete classical \emph{ideal} measurement which reveals everything about the system while touching it so lightly that there is no disturbance at all. An example would be a measurement which simply ``reads off'' the phase space coordinates of a system. In principle, nothing in classical physics prohibits the possibility of this kind of measurement. 

As it reveals the true, underlying properties of the system, the Shannon entropy after such a measurement would be zero --- the expected entropy reduction is maximal, or the measurement is maximally informative in this way. The production of total information from a classical ideal measurement may be considered roughly analogous to the quantum situation of rank-1 MICs which have maximal measurement strength and are informationally complete.

Furthermore, as the true, underlying properties are perfectly distinguishable, we can think of a classical ideal measurement as being perfectly orthogonal; overlap between measurement outcomes in a classical sense implies the measurement is more coarse-grained than is possible in principle. 

Thus, a classical ideal measurement simultaneously produces zero disturbance, maximal information, and is orthonormal. The first two characteristics can immediately be contrasted with Propositions \ref{thm1} and \ref{thm2}, where zero disturbance implies zero, i.e.\ \emph{minimal}, measurement strength. 

While a quantum projective measurement can be orthonormal, this case is not analogous to a classical ideal because a projective measurement falls hopelessly short of revealing ``everything about a system''; only informationally complete measurements can serve this purpose. Among these, the MICs do so with the fewest effects. Just as a classical ideal measurement does not have more outcomes than possible states of the world, a MIC has no more effects than are necessary to capture the full information contained in a quantum state. Pushing as close as possible to the classical extremes among valid MICs to maximize measurement strength while minimizing disturbance and orthogonality, we find that SICs are the quantum analog of a classical ideal measurement. 

SICs have previously been found to be optimal measurements for many purposes, both applied and conceptual~\cite{TFR+21,AFZ15,SIC17}, ranging from leading to the best average fidelity in quantum tomography~\cite{Scott06,Zhu11} to having applications in entanglement witnessing~\cite{SAZ+18}, quantum key distribution~\cite{Ren04}, and random number generation~\cite{AAS+16}. In a similar vein to our finding, SICs have also been seen to be ``most classical'' in that they furnish a probabilistic representation of quantum mechanics in which the Born rule takes a form as close as possible to the classical law of total probability~\cite{DFS20a,DFS20}.

\section{\label{app:concl} Conclusions}
We define the intrinsic disturbance of a measurement with L\"uders rule updating to be a state-independent upper bound of the square of the Hilbert-Schmidt distance between an input state and its expected post-measurement state. 
Disturbance can be decomposed into two parts: one part is related to
the total skew information gain (measurement strength) and the other captures how far the POVM effects differ from an orthonormal basis (orthogonality). 

Despite having maximal measurement strength and being perfectly orthonormal, a von Neumann measurement's disturbance is elevated above many measurements which are closer to being informationally complete. For low rank MICs, on the other hand, disturbance is largely determined by orthogonality, which is restricted. The lower bounds of orthogonality and the corresponding
measurements are found in many important cases.
We single out SICs as the 
analog of a classical ideal measurement as they minimize disturbance among MICs with maximal measurement strength. Our relation could reveal new physical insights into the meaning of measurement and potentially guide the construction of measurements when reducing disturbance is needed. 

Various other information-disturbance trade-offs have been proposed in the past based on other quantifiers, for example von Neumann entropy \cite{FP96,F98,FJ01,Mac06,Luo10,ZZY16}, fidelity \cite{FP96,F98,Ban01,Mac06}, and trace related norms \cite{FJ01,SKU16,BGD+20} (e.g., the Brukner-Zeilinger information and the Fisher information). Our relation is most notable for bringing
orthogonality to the foreground and for identifying a previously unknown role for the total skew information. Accordingly, these rarely used quantities deserve further study. 


\begin{acknowledgments}
We thank two anonymous referees whose feedback led to a significant rewrite and improvements throughout. We are also grateful to Christopher Fuchs and Blake Stacey for valuable comments on the manuscript.
This work is supported by the ORIC project 
from TEEP of Tsinghua University.
\end{acknowledgments}

\appendix

\section{Information of a state}
\label{sec:appa}

Brukner and Zeilinger \cite{BZ01,BZ99} proposed an alternative to von Neumann entropy for quantum systems, now called the Brukner--Zeilinger invariant information of a system $\rho$~\cite{BZ01,BZ99,BZ09}
\begin{equation}
    I_{\rm BZ}(\rho) =  \tr \rho^2 - \frac{1}{d}\;.
\end{equation}

Luo \cite{Luo07} found a connection between $I_{\rm BZ}(\rho)$ and the variance. Traditional variance defined for a pair of a quantum state $\rho$ and an observable $X$ is
\begin{equation}
    V(\rho, X) = \tr \rho X^2 - (\tr \rho X)^2\;.
\end{equation}
Let $\{X_i\}_{i=1}^{d^2}$ be an orthonormal basis of $L(\mathcal{H}_d)$. We can define
\begin{equation}
    V(\rho) :=\sum_i V(\rho, X_i) = d - \tr \rho^2\;.
    \label{Vrho}
\end{equation}
Note that the maximum of $V(\rho)$ is $V_{\max} = d-1/d$, so
we have
\begin{equation}
   I_{\rm BZ}(\rho) =   V_{\max} - V(\rho)\;.
   \label{IBZ}
\end{equation}
The quantity $V(\rho, X)$ is usually interpreted as the uncertainty of $X$ in the state $\rho$. Then, $V(\rho)$ is a quantity summarizing the total uncertainty of $\rho$ which is independent of observables.
Variance and entropy quantify uncertainty. The more uncertain a system is, the less information the system contains. So, $V_{\max} - V(\rho)$ could be a measure of information, which surprisingly coincides with $I_{\rm BZ}(\rho)$. 
Similarly, from an information measure, we can define corresponding entropy-like measure as $I_{\rm BZ \max} - I_{\rm BZ}(\rho)$. We find
\begin{equation}
    S^{(2)}_1 = I_{\rm BZ \max} - I_{\rm BZ}(\rho),
\end{equation}
which is an often-used general entropy called purity \cite{SKW+18,LS17}.

In Ref.~\cite{Luo06,Li11}, the Wigner-Yanase skew information \cite{WY63},
\begin{equation}
    I(\rho, X) = -\frac{1}{2}\tr [\sqrt{\rho},X]^2\;,
\end{equation}
replaces the variance $V(\rho, X)$, giving
\begin{equation}
    I_{\rm WY}(\rho) := \sum_i I(\rho, X_i) = d -(\tr \sqrt{\rho})^2\;.
\end{equation}
Because $\sqrt{U\rho U^{\dagger}} = U\sqrt{\rho} U^{\dagger}$ for any unitary $U$, $I_{\rm WY}(\rho)$ is also invariant under unitary transformations. Similarly, the entropy like value corresponding to $I_{\rm WY}$ is 
\begin{equation}
    S^{(1/2)}_2 = I_{\rm WY \max} - I_{\rm WY}(\rho)\;,
\end{equation}
which is the general entropy we use in Sec.~\ref{sec:reso}.

Some observables which commute with additive conserved quantities (e.g., energy, components of the linear and angular momenta,
electric charge), can be measured more easily than others (i.e., the Wigner-Araki-Yanase theorem) \cite{AY60}. 
Wigner and Yanase argued that there should be functions in the quantum theory to measure our knowledge with respect to these more accessible quantities. They proposed $I(\rho, X)$ to fulfill this task \cite{WY63}. Also, Wigner and Yanase gave five postulates on the information content of a quantum system, which can all be satisfied by $I(\rho, X)$. Similar to the idea of Eqs.~(\ref{Vrho}) and (\ref{IBZ}), $I_{\rm WY}(\rho)$, the total intrinsic uncertainty or information of a system, is obtained by averaging over any basis of the operator space. 

\section{Proof of Theorem 1}\label{app:proof}
First, we need a Lemma.
\begin{lem}
For two nonzero positive semidefinite 
operators $A$ and $B$,
\begin{eqnarray}
\tr(\sqrt{A}\sqrt{B}) \geq 
\frac{\tr AB}{\sqrt{\tr A \tr B}}\;,
\end{eqnarray}
with equality iff $A$ and $B$ are rank-1 or
orthogonal.
\end{lem}

\begin{proof}
In the orthonormal basis 
formed by $A$'s eigenvectors, 
$A$ has the matrix representation
$a D_A = a \text{diag}(\lambda^{A}_1,...,\lambda^{A}_d)$,
while $B$ can be represented by
$bUD_B U^{\dagger}$ with 
$D_B =  \text{diag}(\lambda^{B}_1,...,\lambda^{B}_d)$
and $U = (u_{ij})_{d\times d}$ a unitary matrix
, where $a = \tr A$ and $b = \tr B$.

So we can see
that 
$0 \leq \lambda_i^{\Gamma} \leq 1~(i=1,...,d;\Gamma = A,B)$, $\sum_i\lambda_i^{\Gamma} =1~(\Gamma = A,B)$
and thus
\begin{eqnarray}
\tr(\sqrt{A}\sqrt{B}) &&= 
\sqrt{ab} 
\tr(\sqrt{D_A}U\sqrt{D_B}U^{\dagger})
\nonumber
\\
&&=\sqrt{ab} \sum_{ij} \sqrt{\lambda_i^A} 
\sqrt{\lambda_j^B} u_{ij}^* u_{ij}
\nonumber
\\
&& \geq \sqrt{ab} \sum_{ij}
\lambda_i^A \lambda_j^B u_{ij}^* u_{ij}
\nonumber
\\
&& = \tr (AB) /\sqrt{\tr A \tr B}\;.
\end{eqnarray}
The inequality will be saturated if and only if 
$\tr(AB) = 0$ or
there is a pair $(i,j)$ such that $\lambda_i^A \lambda_j^B
=1$, which is equivalent to the statement in Lemma 1.
\end{proof}

In the following, we use the theory of 
majorization~\cite{Horn91,Mar11,Bha13}.
Let $\vec{x}$ and $\vec{y}$ be two real
vectors with $N$ elements and let $x^{\downarrow}$ and $y^{\downarrow}$ denote the vectors sorted in nonincreasing order.
Then we say $\vec{x}$ weakly majorizes
$\vec{y}$ from below, denoted $\vec{x} \succ_w \vec{y}$, if
\begin{eqnarray}
\sum_{i=1}^k x_i^{\downarrow} \geq 
\sum_{i=1}^k y_i^{\downarrow}~\text{for}~k=1,...,N\;.
\end{eqnarray}
If the last inequality is an equality, we say that
$\vec{x}$ majorizes $\vec{y}$, denoted 
$\vec{x} \succ \vec{y}$. 

Now, we can prove Theorem \ref{thm3} as follows
\begin{proof}
Let $\lambda(A)$ and $s(A)$ denote the 
vector of eigenvalues and singular values of  
the matrix $A$
in non-increasing order.
By the definition of an EA-POVM
, it can be directly worked out that
\begin{eqnarray}
\lambda(S_{n\rm EA}) = \left(1,\frac{d-1}{n-1},
...,\frac{d-1}{n-1}
\right)\;.
\end{eqnarray}

For any POVM $\E_{n}=\{E_i \}$ which has $n$ effects,
we can define a vector 
$\vec{v} = (\sqrt{e_1},...,\sqrt{e_n})^T/\sqrt{d}$,
where $e_i = \tr E_i$.
Suppose $\lambda(S_n) = (\lambda_1,...,\lambda_n)$,
then
\begin{eqnarray}
\lambda_1 &&\geq 
\frac{\vec{v}^T S_n \vec{v}}{\vec{v}^T\vec{v}}
\nonumber
\\
&&=\frac{1}{d}
\sum_{l,m}\tr(\sqrt{E_l}\sqrt{E_m})\sqrt{e_l}\sqrt{e_m}
\nonumber
\\
&&\geq \sum_{l,m}\tr(E_lE_m)/d = 1\;.
\end{eqnarray}
The last inequality is secured by Lemma 1.
Therefore,
\begin{eqnarray}
\lambda(S_n)\; &&\succ
\left(\lambda_1,\frac{d-\lambda_1}{n-1},
...,\frac{d-\lambda_1}{n-1}\right)
\nonumber
\\
&&\succ \left(1,\frac{d-1}{n-1},
...,\frac{d-1}{n-1}\right) = \lambda(S_{n \rm EA})\;.
\end{eqnarray}
Introduce the vector $l_n = (1,...,1)$ which has $n$
entries, it can be established for any real $\gamma$ that
\begin{eqnarray}
 \gamma l_n - \lambda(S_n)
\succ \gamma l_n - \lambda(S_{n\rm EA})\;.
\end{eqnarray}
From Theorem \uppercase\expandafter{\romannumeral2}.3.3 in Ref.~\cite{Bha13},
it holds for two real vectors $(x_1,...,x_N)$ and $(y_1,...,y_N) \in \mathbb{R}^N$ that,
\begin{eqnarray}
&&(x_1,...,x_N) \succ (y_1,...,y_N)
\nonumber\\
\Rightarrow &&(|x_1|,...,|x_N|) \succ_w (|y_1|,...,|y_N|)\;.
\end{eqnarray}
Therefore, we readily draw the conclusion,
\begin{eqnarray}
s(\gamma\1 - S_n)\succ_w
s(\gamma\1 - S_{n\rm EA})\;.
\end{eqnarray}
Since every unitarily invariant norm
is monotone with respect to 
the partial order induced by 
weak majorization of the vector of 
sigular values (see Corollary 3.5.9 in Ref.~\cite{Horn91}), the above formula 
is equivalent to our desired result
Eq.~(\ref{T3}). 
\end{proof}

\section{Orthogonality when $n\notin [d,d^2]$}
\label{sec:appb}
What are the minimal orthogonality POVMs when $n \notin [d,d^2]$, where an EA-POVM cannot exist? Answering this is left to future studies. Here we make some initial observations.

\subsection{Fewer effects}
Suppose $n<d$. One can always find 
POVMs with mutually orthogonal, although not normalized, effects.

Let $\{|\phi_i\y\}$ be an orthonormal basis of $\mathcal{H}_d$ and $d \div n = a \cdots b$ where 
``$\div$" stands for integer division and $b$ is the 
remainder. We can construct a POVM $\{E_i\}_{i=1}^n$ in the following way:
the matrix representations of the effects in the basis
$\{|\phi_i\y\}$ are diagonal and given by
\begin{eqnarray}
(E_i)_{lm} &&= 0,~l \neq m\;,
\nonumber\\
 (E_i)_{lm}&&= 0~\text{or}~1,~l = m\;,
\end{eqnarray}
where there are $a$ nonzero entries in $E_i~(i<n-1)$ and $a+b$ nonzero entries in $E_n$. For example, if $d=3$ and $n=2$, the first effect would have one nonzero element, $(E_1)_{11}=1$, and the second effect would have two, $(E_2)_{22}=1$ and $(E_2)_{33}=1$. $E_i$ does not overlap with $E_j$ if $i\neq j$. Thus, $\sum_i E_i = I$ and the $S$ matrix of this POVM is given by 
\begin{eqnarray}
S = \text{diag}(a,...,a,a+b)\;.
\label{Sn<d}
\end{eqnarray}
The measurement strength $R$, orthogonality $O$, and disturbance $D$ of
this POVM are 
\begin{eqnarray}
R=a(d+b)\frac{n-1}{d}\;,
\end{eqnarray}
\begin{eqnarray}
O=(a-1)(d+b)-n(a-1)+b^2\;,
\end{eqnarray}
\begin{eqnarray}
D= d(an-b-a) + b(an-a+d)\;.
\end{eqnarray}

Although for $n<d$, there are POVMs such as the one we just constructed whose operators are pairwise orthogonal, their orthogonality $O$ may not be zero as $O$ is actually a measure of \emph{nonorthonormality}. From Eq.~(\ref{Sn<d}), we conjecture the above POVM is 
the closet to being ``orthonormal''.

\begin{conj}
The above constructed POVM has 
minimal orthogonality $O$ among POVMs having $n(n<d)$ nonvanishing effects.
\end{conj}

\subsection{More effects}
Suppose $n>d^2$ and let $\E_{n,r}=\{E_i\}_{i=1}^n$ be a POVM, where $r$ labels the dimension of the space spanned by the operators $\{\sqrt{E_i}\}$, that is,  $r=\text{dim}\left(\text{span}\left(\{\sqrt{E_i}\}\right)\right)$. Naturally, $r \leq d^2$. If $r \geq d$, the following inequality holds:
\begin{eqnarray}
O(\E_{n,r}) \geq O(\E_{r\rm EA}) + n - r\;,
\label{BO}
\end{eqnarray}
where $\E_{r\rm EA}$ is an EA-POVM with $r$ effects.

\emph{Proof.} 
Let $S_{n,r}$ be the $S$ matrix of $\E_{n,r}$. Following the same arguments in the proof of Theorem \ref{thm3}, we have
\begin{eqnarray}
\lambda(S_{n,r}) &&\succ
\left(\lambda_1,\frac{d-\lambda_1}{r-1},
...,\frac{d-\lambda_1}{r-1},0...,0\right)
\nonumber
\\
&&\succ \left(1,\frac{d-1}{r-1},
...,\frac{d-1}{r-1},0,...,0\right),
\end{eqnarray}
where $\lambda_1 \geq 1$ is the greatest eigenvalue of 
$S_{n,r}$ and there are $n-r$ zero elements in $\lambda(S_{n,r})$. 
Mirroring Theorem \ref{thm3}, it can be shown that 
\begin{eqnarray}
 \lambda(\1 - S_{n,r})
\succ \left(0,\frac{r-d}{r-1},
...,\frac{r-d}{r-1},1,...,1\right),
\end{eqnarray}
and
\begin{eqnarray}
 s(\1 - S_{n,r})
\succ_w \left(0,\frac{r-d}{r-1},
...,\frac{r-d}{r-1},1,...,1\right).
\label{B9}
\end{eqnarray}
As we choose the Hilbert-Schmidt norm to define $O$, Eq.~(\ref{BO}) can be directly deduced from the fact
\begin{eqnarray}
 \|\1 - &&S_{n,r} \| \geq
 \nonumber\\
 &&\left\|\text{diag}\left(0,\frac{r-d}{r-1},
...,\frac{r-d}{r-1},1,...,1\right) \right\|\;,
\end{eqnarray}
which follows from Eq.~(\ref{B9}).
$\qed$

If $r=d^2$ and $\E_{n,r}$ is rank-1, Eq.~(\ref{BO}) 
gives $O(\E_{n,r}) \geq O(\E_{\rm SIC})$ and $D(\E_{n,r}) \geq D(\E_{\rm SIC})$ where $\E_{\rm SIC}$ is a SIC. Example 4 in Sec.~\ref{sec:dis} shows that the latter inequality
can be saturated by POVMs which form a 2-design.

\begin{conj}
Among POVMs with $n(n>d^2)$ effects, those which form a 2-design have the minimal orthogonality $O$.
\end{conj}


\bibliography{apssamp}

\end{document}